\documentclass[11pt]{article}
\usepackage{graphicx}
\usepackage{amsmath, amssymb}

\usepackage[noline,noend,linesnumberedhidden,boxed]{algorithm2e}
\DontPrintSemicolon
\SetArgSty{textnormal}
\SetAlCapSkip{1ex}

\newcommand{\qed}{\rule{7pt}{7pt}}
\newcommand{\proofof}[1]{\smallskip\par\noindent{\sl Proof of #1}:\enspace}
\newenvironment{proof}{\noindent {\it Proof\/}:}{$\qed$ \medskip}

\newtheorem{theorem}{Theorem}[section]
\newtheorem{lemma}[theorem]{Lemma}

\def\({\left(}
\def\){\right)}
\newcommand{\OPT}{\mathrm{OPT}}
\newcommand{\SAT}{\mathrm{SAT}}
\newcommand{\UNSAT}{\mathrm{UNSAT}}
\newcommand{\LP}{\mathrm{LP}}

\begin{document}

\title{On Some Recent MAX SAT Approximation Algorithms}
\author{Matthias Poloczek\thanks{Address: Institute of Computer Science, Goethe University, Frankfurt am Main, Germany.
Email: \texttt{matthias@thi.cs.uni-frankfurt.de}.}
 \and David P.\ Williamson\thanks{Address: School
of Operations Research and Information Engineering, Cornell
University, Ithaca, NY, USA.  Email: {\tt dpw@cs.cornell.edu}.
Supported in part by NSF grant CCF-1115256.}
\and Anke van Zuylen\thanks{Address: Department of Mathematics, College of William and Mary, Williamsburg, VA, USA.   Email: {\tt anke@wm.edu.}}}

\date{}
\maketitle
%% CHANGED
\thispagestyle{empty}

\begin{abstract}
Recently a number of randomized $\frac{3}{4}$-approximation algorithms for MAX SAT have been proposed that all work in the same way: given a fixed ordering of the variables, the algorithm makes a random assignment to each variable in sequence, in which the probability of assigning each variable true or false depends on the current set of satisfied (or unsatisfied) clauses.  To our knowledge, the first such algorithm was proposed by Poloczek and Schnitger \cite{PoloczekS11}; Van Zuylen \cite{VanZuylen11} subsequently gave an algorithm that set the probabilities differently and had a simpler analysis.  She also set up a framework for deriving such algorithms.  Buchbinder, Feldman, Naor, and Schwartz \cite{BuchbinderFNS12}, as a special case of their work on maximizing submodular functions, also give a randomized $\frac{3}{4}$-approximation algorithm for MAX SAT with the same structure as these previous algorithms.  In this note we give a gloss on the Buchbinder et al.\ algorithm that makes it even simpler, and show that in fact it is equivalent to the previous algorithm of Van Zuylen.   We also show how it extends to a deterministic LP rounding algorithm; such an algorithm was also given by Van Zuylen \cite{VanZuylen11}.
\end{abstract}

%% CHANGE
%%
%\setcounter{page}{0}
%\newpage

\section{Introduction}

The maximum satisfiability problem (MAX SAT) is a fundamental problem in discrete optimization.  In the problem we are given $n$ boolean variables $x_1,\ldots,x_n$ and $m$ clauses that are conjunctions of the variables or their negations.  With each clause $C_j$, there is an associated weight $w_j \geq 0$.  We say a clause is {\em satisfied} if one of its positive variables is set to true or if one of its negated variables is set to false.  The goal of the problem is to find an assignment of truth values to the variables so as to maximize the total weight of the satisfied clauses.  The problem is NP-hard via a trivial reduction from satisfiability.

We say we have an $\alpha$-approximation algorithm for MAX SAT if we have a polynomial-time algorithm that computes an assignment whose total weight of satisfied clauses is at least $\alpha$ times that of an optimal solution; we call $\alpha$ the performance guarantee of the algorithm.  A randomized $\alpha$-approximation algorithm is a randomized polynomial-time algorithm such that the expected weight of the satisfied clauses is at least $\alpha$ times that of an optimal solution.  The 1974 paper of Johnson \cite{Johnson74}, which introduced the notion of an approximation algorithm, also gave a $\frac{1}{2}$-approximation algorithm for MAX SAT.  This algorithm was later shown to be a $\frac{2}{3}$-approximation algorithm by Chen, Friesen, and Zheng \cite{ChenFZ99} (see also the simpler analysis of Engebretsen \cite{Engebretsen04}).  Yannakakis \cite{Yannakakis94} gave the first $\frac{3}{4}$-approximation algorithm for MAX SAT; it uses network flow and linear programming computation as subroutines.  Goemans and Williamson \cite{GoemansW94} subsequently showed how to use randomized rounding of a linear program to obtain a $\frac{3}{4}$-approximation algorithm for MAX SAT.  Subsequent approximation algorithms which use semidefinite programming have led to still better performance guarantees.

In 1998, Williamson \cite[p.\ 45]{Williamson98} posed the question of whether it is possible to obtain a $\frac{3}{4}$-approximation algorithm for MAX SAT without solving a linear program.  This question was answered positively in 2011 by Poloczek and Schnitger \cite{PoloczekS11}.  They give a randomized algorithm with the following particularly simple structure: given a fixed ordering of the variables, the algorithm makes a random assignment to each variable in sequence, in which the probability of assigning each variable true or false depends on the current set of satisfied (or unsatisfied) clauses.  Subsequently, Van Zuylen \cite{VanZuylen11} gave an algorithm with the same structure that set the probabilities differently and had a simpler analysis.  She also set up a framework for deriving such algorithms.  In 2012, Buchbinder, Feldman, Naor, and Schwartz \cite{BuchbinderFNS12}, as a special case of their work on maximizing submodular functions, also gave a randomized $\frac{3}{4}$-approximation algorithm for MAX SAT with the same structure as these previous algorithms\footnote{In the extended abstract of \cite{BuchbinderFNS12}, the authors claim the MAX SAT result and omit the proof, but it is not difficult to reconstruct the proof from the rest of the paper.}.
Poloczek \cite{Poloczek11} gives evidence that the randomization is necessary for this style of algorithm by showing that a deterministic algorithm that sets the variables in order (where the next variable to set is chosen adaptively) and uses a particular set of information about the clauses cannot achieve performance guarantee better than $\frac{\sqrt{33}+3}{12} \approx .729$.  However, Van Zuylen \cite{VanZuylen11} shows that it is possible to give a deterministic $\frac{3}{4}$-approximation algorithm with the same structure given a solution to a linear programming relaxation.

The goal of this note is to give an interpretation of the Buchbinder et al.\ MAX SAT algorithm that we believe is conceptually simpler than the one given there.  We also restate the proof in terms of our interpretation.  We further show that the Buchbinder et al.\ algorithm is in fact equivalent to the previous algorithm of Van Zuylen.  We extend the algorithm and analysis to a deterministic LP rounding algorithm.

Here we give the main idea of our perspective on the algorithm.  Consider greedy algorithms that set the variables $x_i$ in sequence. A natural greedy algorithm sets $x_i$ to true or false depending on which increases the total weight of the satisfied clauses by the most.  An alternative to this algorithm would be to set each $x_i$ so as to increase the total weight of the clauses that are not yet {\em unsatisfied} given the setting of the variable (a clause is unsatisfied if all the variables of the clause have been set and their assignment does not satisfy the clause).  The algorithm is in a sense a randomized balancing of these two algorithms.  It maintains a bound that is the average of  two numbers, the total weight of the clauses satisfied thus far, and the total weight of the clauses that are not yet unsatisfied.  For each variable $x_i$, it computes the amount by which the bound will increase if $x_i$ is set true or false; one can show that the sum of these two quantities is always nonnegative.  If one assignment causes the bound to decrease, the variable is given the other assignment (e.g.\ if assigning $x_i$ true decreases the bound, then it is assigned false).  Otherwise, the variable is set randomly with a bias towards the larger increase.

This note is structured as follows.  Section \ref{sec:notation} sets up some notation we will need.  Section \ref{sec:alg} gives the randomized $\frac{3}{4}$-approximation algorithm and its analysis.  Section \ref{sec:lp} extends these to a deterministic LP rounding algorithm.  Section \ref{sec:anke} explains how the algorithm is equivalent to the previous algorithm of Van Zuylen.  We conclude with some open questions in Section \ref{sec:conc}.

\section{Notation}
\label{sec:notation}

We assume a fixed ordering of the variables, which for simplicity will be given as $x_1,x_2,\ldots,x_n$.  As the algorithm proceeds, it will sequentially set the variables; let $S_i$ denote some setting of the first $i$ variables.  Let $W= \sum_{j=1}^m w_j$ be the total weight of all the clauses.  Let $\SAT_i$ be the total weight of clauses satisfied by $S_i$, and let $\UNSAT_i$ be the total weight of clauses that are unsatisfied by $S_i$; that is, clauses that only have variables from $x_1,\ldots,x_i$ and are not satisfied by $S_i$.   Note that $\SAT_i$ is a lower bound on the total weight of clauses satisfied by our final assignment $S_n$ (once we have set all the variables); furthermore, note that $W-\UNSAT_i$ is an upper bound on the total weight of clauses satisfied by our final assignment $S_n$.  We let $B_i  = \frac{1}{2}(\SAT_i + (W-\UNSAT_i))$ be the midpoint between these two bounds; we refer to it simply as the bound on our partial assignment $S_i$.  For any assignment $S$ to all of the variables, let $w(S)$ represent the total weight of the satisfied clauses.  Then we observe that for the assignment $S_n$, $w(S_n) = \SAT_n = W-\UNSAT_n$, so that $w(S_n) = B_n$.  Furthermore, $\SAT_0 = 0$ and $\UNSAT_0 = 0$, so that $B_0 = \frac{1}{2}W$.

Note that our algorithm will be randomized, so that $S_i$, $\SAT_i$, $\UNSAT_i$, and $B_i$ are all random variables.

\section{The Algorithm and its Analysis}
\label{sec:alg}

The goal of the algorithm is at each step to try to increase the bound;  that is, we would like to set $x_i$ randomly so as to increase $E[B_i - B_{i-1}]$.  We let $t_i$ be the value of $B_i - B_{i-1}$ in which we set $x_i$ true, and $f_i$ the value of $B_i - B_{i-1}$ in which we set $x_i$ false.  Note that the expectation is conditioned on our previous setting of the variables $x_1,\ldots,x_{i-1}$, but we omit the conditioning for simplicity of notation.  We will show momentarily that $t_i + f_i \geq 0$.  Then the algorithm is as follows.  If $f_i \leq 0$, we set $x_i$ true; that is, if setting $x_i$ false would not increase the bound, we set it true.  Similarly, if $t_i \leq 0$ (setting $x_i$ true would not increase the bound) we set $x_i$ false.  Otherwise, if either setting $x_i$ true or false would increase the bound, we set $x_i$ true with probability $\frac{t_i}{t_i + f_i}$.

\begin{lemma} \label{lem:geq} For $i = 1,\ldots,n$,
$$t_i + f_i \geq 0.$$
\end{lemma}

\begin{proof}
We note that any clause that becomes unsatisfied by $S_{i-1}$ and setting $x_i$ true must be then be satisfied by setting $x_i$ false, and similarly any clause that becomes unsatisfied by $S_{i-1}$ and setting $x_i$ false must then be satisfied by setting $x_i$ true.  Let $\SAT_{i,t}$ be the clauses that are satisfied by setting $x_i$ true given the partial assignment $S_{i-1}$, and $\SAT_{i,f}$ be the clauses satisfied by setting $x_i$ false given the partial assignment $S_{i-1}$.  We define $\UNSAT_{i,t}$ ($\UNSAT_{i,f}$) to be the clauses unsatisfied by $S_{i-1}$ and $x_i$ set true (respectively false).  Our observation above implies that $\SAT_{i,f} - \SAT_{i-1} \geq \UNSAT_{i,t} - \UNSAT_{i-1}$ and $\SAT_{i,t} - \SAT_{i-1} \geq \UNSAT_{i,f} - \UNSAT_{i-1}.$

Let $B_{i,t} = \frac{1}{2}(\SAT_{i,t} + (W - \UNSAT_{i,t}))$ and $B_{i,f} = \frac{1}{2}(\SAT_{i,f} + (W - \UNSAT_{i,f}))$.  Then $t_i = B_{i,t} - B_{i-1}$ and $f_i = B_{i,f} - B_{i-1}$; our goal is to show that $t_i + f_i \geq 0$, or
$$\frac{1}{2}(\SAT_{i,t} + (W - \UNSAT_{i,t})) + \frac{1}{2}(\SAT_{i,f} + (W - \UNSAT_{i,f})) - \SAT_{i-1} - (W - \UNSAT_{i-1}) \geq 0.$$  Rewriting, we want to show that $$\frac{1}{2}(\SAT_{i,t} - \SAT_{i-1}) + \frac{1}{2}(\SAT_{i,f} - \SAT_{i-1}) \geq \frac{1}{2}(\UNSAT_{i,f} - \UNSAT_{i-1}) + \frac{1}{2}(\UNSAT_{i,t} - \UNSAT_{i-1}),$$ and this follows from the inequalities of the previous paragraph.
\end{proof}

Let $x^*$ be a fixed optimal solution. Following both Poloczek and Schnitger, and Buchbinder et al., given a partial assignment $S_i$, let $\OPT_i$ be the assignment in which variables $x_1,\ldots,x_i$ are set as in $S_i$, and $x_{i+1},\ldots,x_n$ are set as in $x^*$.  Thus if $\OPT$ is the value of an optimal solution, $w(\OPT_0) = \OPT$, while $w(\OPT_n) = w(S_n)$.

The following lemma is at the heart of both analyses (see Section 2.2 in Poloczek and Schnitger \cite{PoloczekS11} and Lemma III.1 of Buchbinder et al.\ \cite{BuchbinderFNS12}).

\begin{lemma} \label{lem:main} For $i=1,\ldots,n$, the following holds:
$$E[w(\OPT_{i-1}) - w(\OPT_i)] \leq E[B_i - B_{i-1}].$$
\end{lemma}

Before we prove the lemma, we show that it leads straightforwardly to the desired approximation bound.

\begin{theorem} \label{thm:34}
$$E[w(S_n)] \geq \frac{3}{4} \OPT.$$
\end{theorem}

\begin{proof}
We sum together the inequalities from the lemma, so that
$$\sum_{i=1}^{n} E[w(\OPT_{i-1}) - w(\OPT_i)] \leq  \sum_{i=1}^{n} E[B_i - B_{i-1}].$$
Using the linearity of expectation and telescoping the sums, we get
$$E[w(\OPT_0) - w(\OPT_n)] \leq E[B_n] - E[B_0].$$
Thus
$$\OPT - E[w(S_n)] \leq E[w(S_n)] - \frac{1}{2}W,$$
or
$$\OPT + \frac{1}{2} W \leq 2 E[w(S_n)],$$
or
$$\frac{3}{4} \OPT \leq E[w(S_n)]$$ as desired, since $\OPT \leq W$. \end{proof}

The following lemma is the key insight of proving the main lemma, and is the randomized balancing of the two greedy algorithms mentioned in the introduction. The expectation bound holds whether $x^*_i$ is true or false.

\begin{lemma} \label{lem:opti}
$$E[w(\OPT_{i-1}) - w(\OPT_i)] \leq \max\left(0,\frac{2t_i f_i}{t_i + f_i}\right).$$
\end{lemma}

\begin{proof}
Assume for the moment that $x^*_i$ is set false; the proof is analogous if $x_i^*$ is true. We claim that if $x_i$ is set true while $x^*_i$ is false, then $w(\OPT_{i-1}) - w(\OPT_i) \leq 2f_i$.  If $f_i \leq 0$, then we set $x_i$ true and the lemma statement holds given the claim.  If $t_i \leq 0$, we set $x_i$ false; then the assignment $\OPT_i$ is the same as $\OPT_{i-1}$ so that  $w(\OPT_i)- w(\OPT_{i-1}) = 0$ and the lemma statement again holds.  Now assume both $f_i > 0$ and $t_i > 0$.  We set $x_i$ false with probability $f_i/(t_i + f_i)$, so that again $w(\OPT_i)- w(\OPT_{i-1}) = 0$.  We set $x_i$ true with probability $t_i/(t_i + f_i)$.    If the claim holds then the lemma is shown, since
$$E[w(\OPT_{i-1}) - w(\OPT_i)] \leq \frac{f_i}{t_i + f_i} \cdot 0 + \frac{t_i}{t_i + f_i} \cdot 2f_i = \frac{2t_i f_i}{t_i + f_i}.$$

If $x_i$ is set true while $x^*_i$ is false,  $\OPT_{i-1}$ differs from $\OPT_i$ precisely by having the $i$th variable set false, and $w(\OPT_{i-1}) - w(\OPT_i)$ is the difference in the weight of the satisfied clauses made by flipping the $i$th variable from true to false. Since both assignments have the first $i-1$ variables set as in $S_{i-1}$, they both satisfy at least $\SAT_{i-1}$ total weight, so the increase of flipping the $i$th variable from true to false is at most $\SAT_{i,f} - \SAT_{i-1}$.  Additionally, both assignments leave at least $\UNSAT_{i-1}$ total weight of clauses unsatisfied, so that flipping the $i$th variable from true to false leaves at least $\UNSAT_{i,f} - \UNSAT_{i-1}$ additional weight unsatisfied.  In particular, flipping the $i$th variable will unsatisfy an additional $\UNSAT_{i,f} - \UNSAT_{i-1}$ weight of the clauses that only have variables from $x_1,\ldots,x_i$ and may unsatisfy additional clauses as well.  Thus, if $x_i$ is set true and $x^*_i$ is false,
\begin{align*}
w(\OPT_{i-1}) - w(\OPT_i) & \leq
(\SAT_{i,f}-\SAT_{i-1}) -(\UNSAT_{i,f} - \UNSAT_{i-1})\\
& = (\SAT_{i,f}+(W-\UNSAT_{i,f})) - (\SAT_{i-1} + (W - \UNSAT_{i-1})) \\
& =  2(B_{i,f} - B_{i-1}) = 2f_i.
\end{align*}
\end{proof}

Now we can prove the main lemma.

\proofof{Lemma \ref{lem:main}}
If either $t_i \leq 0$ or $f_i \leq 0$, then by Lemma \ref{lem:opti}, we set $x_i$ deterministically so that the bound does not decrease and $B_i - B_{i-1} \geq 0.$ Since then $t_i f_i \leq 0$, by Lemma \ref{lem:opti} $$E[w(\OPT_{i-1}) - w(\OPT_i)] \leq \max(0,2t_if_i/(t_i + f_i)) \leq 0,$$ and the inequality holds.

If both $t_i, f_i > 0$,  then
\begin{align*}
E[B_i - B_{i-1}] & = \frac{t_i}{t_i + f_i} [B_{i,t} - B_{i-1}] + \frac{f_i}{t_i + f_i} [B_{i,f} - B_{i-1}] \\
& = \frac{t_i^2 + f_i^2}{t_i + f_i},
\end{align*}
while by Lemma \ref{lem:opti}
$$E[w(\OPT_{i-1}) - w(\OPT_i)] \leq \frac{2t_i f_i}{t_i + f_i}.$$  Therefore in order to verify the inequality, we need to show that when $t_i, f_i > 0$,
$$ \frac{2t_i f_i}{t_i + f_i} \leq \frac{t_i^2 + f_i^2}{t_i + f_i},$$
which follows since $t_i^2 + f_i^2 - 2t_if_i = (t_i - f_i)^2 \geq 0.$
\qed

\section{Another Deterministic LP Rounding Algorithm}
\label{sec:lp}

We can now take essentially the same algorithm and analysis, and use it to obtain a deterministic LP rounding algorithm.

We first give the standard LP relaxation of MAX SAT.  It uses decision variables $y_i \in \{0,1\}$, where $y_i = 1$ corresponds to $x_i$ being set true, and $z_j \in \{0,1\}$, where $z_j = 1$ corresponds to clause $j$ being satisfied.  Let $P_j$ be the set of variables that occur positively in clause $j$ and $N_j$ be the set of variables that occur negatively.  Then the LP relaxation is:
\begin{alignat*}{6}
& \mbox{maximize } & \sum_{j=1}^m w_j z_j \\
& \mbox{subject to } &
\sum_{i \in P_j} y_i + \sum_{i \in N_j} (1 - y_i) & \geq z_j, &
\qquad & \forall C_j = \bigvee_{i \in P_j} x_i \vee
\bigvee_{i \in N_j} \bar x_i,\\
& & 0 \leq y_i & \leq 1, & & i = 1,\ldots,n,\\
& & 0 \leq z_j & \leq 1, & &  j = 1,\ldots,m.
\end{alignat*}
Note that given a setting of $y$, we can easily find the best possible $z$ by setting $z_j = \max(1, \sum_{i \in P_j} y_i + \sum_{i \in N_j} (1 - y_i))$.  Let $\OPT_{\LP}$ be the optimal value of the LP.

Let $y^*$ be an optimal solution to the LP relaxation.  As before, our algorithm will sequence through the variables $x_i$, deciding at each step whether to set $x_i$ to true or false; now the decision will be made deterministically.  Let $B_i$ be the same bound as before, and as before let $t_i$ be the increase in the bound if $x_i$ is set true, and $f_i$ the increase if $x_i$ is set false.  The concept corresponding to $\OPT_i$ in the previous algorithm is $\LP_i$, the best possible solution to the LP for a vector $\hat y$ in which the first $i$ elements are 0s and 1s corresponding to our assignment $S_i$, while the remaining entries are the values of the optimal LP solution $y^*_{i+1}, \ldots, y^*_n$.  Thus $\LP_0= \OPT_{\LP}$  and $\LP_n= w(S_n)$, the weight of our assignment.   We further introduce the notation $\LP_{i,t}$ ($\LP_{i,f}$), which correspond to the best possible solution to the LP for the vector in which the first $i-1$ elements are 0s and 1s corresponding to our assignment $S_{i-1}$, the entries for $i+1$ to $n$ are the values of the optimal LP solution $y^*_{i+1},\ldots,y^*_n$, and the $i$th entry is 1 (0, respectively)  Note that after we decide whether setting $x_i$ true or false, either $\LP_i = \LP_{i,t}$ (if we set $x_i$ true) or $\LP_i=\LP_{i,f}$ (if we set $x_i$ false).

The following lemma is the key to the algorithm and the analysis; we defer the proof.

\begin{lemma} \label{lem:lpkey} For each $i$, $i = 1,\ldots, n$, at least one of the following two inequalities is true:
$$\LP_{i-1} - \LP_{i,t} \leq t_i \mbox{  or  } \LP_{i-1} - \LP_{i,f} \leq f_i.$$
\end{lemma}

The algorithm is then as following: when we consider variable $x_i$, we check whether $\LP_{i-1} - \LP_{i,t} \leq t_i$ and whether $\LP_{i-1} - \LP_{i,f} \leq f_i$.  If the former, we set $x_i$ true (and thus $\LP_i = \LP_{i,t}$); if the latter, we set $x_i$ false (and thus $\LP_i = \LP_{i,f}$).

The following lemma, which is analogous to Lemma \ref{lem:main}, now follows easily.
\begin{lemma} \label{lem:lpmain}
For $i=1,\ldots,n$,
$$\LP_{i-1} - \LP_i \leq B_i - B_{i-1}.$$
\end{lemma}

\begin{proof}
If $\LP_{i-1} - \LP_{i,t} \leq t_i$, then we set $x_i$ true, so that $\LP_{i}  = \LP_{i,t}$ and $B_i - B_{i-1} = B_{i,t}-B_i = t_i$; thus the lemma statement holds.  The case in which $\LP_{i-1} - \LP_{i,f} \leq f_i$ is parallel.
\end{proof}

Given the lemma, we can prove the following.
\begin{theorem} For the assignment $S_n$ computed by the algorithm,
$$w(S_n) \geq \frac{3}{4} \OPT.$$
\end{theorem}

\begin{proof}
As in the proof of Theorem \ref{thm:34}, we sum together the inequalities given by Lemma \ref{lem:lpmain}, so that
$$\sum_{i=1}^n \left(\LP_{i-1} - \LP_i\right) \leq \sum_{i=1}^n \left(B_i - B_{i-1}\right).$$
Telescoping the sums, we get
$$\LP_0 - \LP_n \leq B_n - B_0,$$ or
$$\OPT_{\LP} - w(S_n) \leq w(S_n) - \frac{W}{2}.$$  Rearranging terms, we have
$$w(S_n) \geq \frac{1}{2} \OPT_{\LP} + \frac{W}{4} \geq \frac{3}{4} \OPT,$$ since both $\OPT_{\LP} \geq \OPT$ (since the LP is a relaxation) and $W \geq \OPT$.
\end{proof}

Now to prove Lemma \ref{lem:lpkey}.
\proofof{Lemma \ref{lem:lpkey}}
We claim that $\LP_{i-1} - \LP_{i,t} \leq 2(1 - y_i^*) f_i$ and $\LP_{i-1} - \LP_{i,f} \leq 2y^*_i t_i$.  The lemma statement follows from this claim; to see this, suppose otherwise, and neither statement holds.  Then by the claim and by hypothesis,
$$t_i  < \LP_{i-1} - \LP_{i,t} \leq 2(1 - y_i^*)f_i < 2(1 - y_i^*)(\LP_{i-1} - \LP_{i,f}) \leq 4(1-y^*_i)y^*_i t_i.$$
But this is a contradiction: since $0 \leq y^*_i \leq 1$, $4(1-y^*_i)y^*_i \leq 1$, and the inequality above implies $t_i < t_i$.  So at least one of the two inequalities must hold.

Now to prove the claim.   We observe that $\LP_{i-1} - \LP_{i,t}$ is equal to the outcome of changing an LP solution $y$ from $y_i=1$ to $y_i=y^*_i$; all other entries in the $y$ vector remain the same.  Both $\LP_{i-1}$ and $\LP_{i,t}$ satisfy at least $\SAT_{i-1}$ weight of clauses; any increase in weight of satisfied clauses due to reducing $y_i$ from 1 to $y^*_i$ must be due to clauses in which $x_i$ occurs negatively; if we reduce $y_i$ from 1 to 0, we would get a increase of the objective function of at most $\SAT_{i,f} - \SAT_{i-1}$, but since we reduce $y_i$ from 1 to $y^*_i$, we get $(1 - y^*_i)(\SAT_{i,f} - \SAT_{i-1})$.  Additionally both $\LP_{i-1}$ and $\LP_{i,t}$ have at least $\UNSAT_{i-1}$ total weight of unsatisfied clauses; any increase in the weight of unsatisfied clauses due to reducing $y_i$ from 1 to $y^*_i$ must be due to clauses in which $x_i$ occurs positively; if we reduce $y_i$ from 1 to 0, we would get a decrease in the objective function of at least $\UNSAT_{i,f} - \UNSAT_{i-1}$, but since we reduce $y_i$ from 1 to $y^*_i$ we get $(1 - y^*_i)(\UNSAT_{i,f} - \UNSAT_{i-1})$.  Then we have
$$\LP_{i-1} - \LP_{i,t} \leq (1 - y^*_i)(\SAT_{i,f} - \SAT_{i-1} - (\UNSAT_{i,f} - \UNSAT_{i-1})) \leq 2(1-y^*_i)f_i,$$ where the last inequality follows as in the proof of Lemma \ref{lem:opti}.  The proof that $\LP_{i-1} - \LP_{i,f} \leq 2y^*_i t_i$ is analogous.
\qed

\section{The Van Zuylen Algorithm}
\label{sec:anke}

In this section, we show that Van Zuylen's algorithm is equivalent to the algorithm of Section \ref{sec:alg}.  Van Zuylen's algorithm \cite{VanZuylen11} uses the following quantities to decide how to set each variable $x_i$.  Let $W_i$ be the weight of the clauses that become satisfied by setting $x_i$ true and unsatisfied by setting $x_i$ false, and let $\bar W_i$ be the weight of the clauses satisfied by setting $x_i$ false and unsatisfied by setting $x_i$ true.   Let $F_i$ be the weight of the clauses that are satisfied by setting $x_i$ true, and are neither satisfied nor unsatisfied by setting $x_i$ false; similarly, $\bar F_i$ is the weight of the clauses that are satisfied by setting $x_i$ false, and are neither satisfied nor unsatisfied by setting $x_i$ true.  Then Van Zuylen calculates a quantity $\alpha$ as follows:
$$\alpha = \frac{W_i + F_i - \bar W_i}{F_i + \bar F_i}.$$
Van Zuylen's algorithm sets $x_i$ false if $\alpha \leq 0$, true if $\alpha \geq 1$, and sets $x_i$ true with probability $\alpha$ if $0 < \alpha < 1$.

We observe that in terms of our prior quantities, $W_i + F_i = \SAT_{i,t} - \SAT_{i-1}$, while $\bar W_i = \UNSAT_{i,f} - \UNSAT_{i-1}$.  Then
\begin{align*}
W_i + F_i - \bar W_i & = \SAT_{i,t} - \SAT_{i-1} + (W - \UNSAT_{i,t}) - (W - \UNSAT_{i-1}) \\
& = 2(B_{i,t} - B_{i-1}) = 2t_i.
\end{align*}
Similarly, $\bar W_i + \bar F_i - W_i = 2f_i$.
Furthermore,
\begin{align*}
F_i + \bar F_i & = (F_i + W_i) + (\bar F_i + \bar W_i) - W_i - \bar W_i \\
& = (\SAT_{i,t} - \SAT_{i-1}) + (\SAT_{i,f} - \SAT_{i-1}) \\
& \qquad - (\UNSAT_{i,t} - \UNSAT_{i-1}) - (\UNSAT_{i,f} - \UNSAT_{i-1})\\
& = [(\SAT_{i,t} + (W - \UNSAT_{i,t}) - (\SAT_{i-1} + (W - \UNSAT_{i-1})] \\
& \qquad + [(\SAT_{i,f} + (W - \UNSAT_{i,f}) - (\SAT_{i-1} + (W - \UNSAT_{i-1})] \\
& = 2(B_{i,t} - B_{i-1}) + 2(B_{i,f} - B_{i-1})\\
& = 2(t_i + f_i).
\end{align*}
Thus $\alpha \leq 0$ if and only if $t_i \leq 0$, and in this case $x_i$ is set false.  Also, $\alpha \geq 1$ if and only if
$$W_i + F_i - \bar W_i \geq F_i + \bar F_i$$ or
$$\bar W_i + \bar F_i - W_i \leq 0,$$ which is equivalent to $f_i \leq 0,$ and in this case $x_i$ is set true.  Finally, $0 < \alpha < 1$ if and only if $t_i > 0$ and $f_i > 0$ and in this case, $x_i$ is set true with probability $\alpha = \frac{2t_i}{2(t_i + f_i)} = \frac{t_i}{t_i + f_i}.$  Thus Van Zuylen's algorithm and the algorithm of Section \ref{sec:alg} are equivalent.

\section{Conclusions}
\label{sec:conc}

A natural question is whether there exists a simple deterministic $\frac{3}{4}$-approximation algorithm for MAX SAT that does not require the use of linear programming.  The paper of Poloczek \cite{Poloczek11} rules out certain types of algorithms, but other types might still be possible.  Another question is whether randomization is inherently necessary for the $\frac{1}{2}$-approximation algorithm for (nonmonotone) submodular function maximization of Buchbinder  et al. \cite{BuchbinderFNS12}; that is, can one achieve a deterministic $\frac{1}{2}$-approximation algorithm?  It might be possible to show that given a fixed order of items and a restriction on the algorithm that it must make an irrevocable decision on whether to include an item in the solution set or not, a deterministic $\frac{1}{2}$-approximation algorithm is not possible.  However, it seems that there must be some  reasonable restriction on the number of queries made to the submodular function oracle.

\bibliographystyle{abbrv}
\bibliography{satbib}

\end{document}